\documentclass[11pt]{article}
\usepackage[english]{babel}
\usepackage[utf8]{inputenc}
\usepackage{lmodern}
\usepackage[T1]{fontenc}
\usepackage{marginnote}
\usepackage{amsmath,amssymb,lmodern,amsthm}
\usepackage{mathtools}
\usepackage{graphicx}
\usepackage{mathrsfs}
\usepackage{bm}
\usepackage[breaklinks]{hyperref}
\usepackage{url}
\usepackage{breakurl}
\usepackage{array}
\usepackage{tensor}
\usepackage{pifont}
\usepackage[toc,page]{appendix}
\usepackage[left=2.5cm,right=2.5cm,top=3cm,bottom=3cm]{geometry}
\usepackage{tikz}
\usetikzlibrary{arrows,positioning,matrix}
\usepackage{enumitem}
\usepackage{empheq}
\usepackage{cleveref}
\usepackage{cases}
\usepackage{csquotes}
\usepackage{accents}
\usepackage{setspace}

\theoremstyle{plain}
\usepackage{xargs}
\usepackage{xcolor}
\usepackage{subcaption}
\usepackage[justification=raggedright]{caption}
\usepackage[affil-it]{authblk}
\usepackage{tocloft}

\hypersetup{
	colorlinks=true,
	urlcolor=blue,
	linkcolor={black},
	citecolor={green!40!black}
}
\usepackage{thmtools}
\usepackage[backend=bibtex,style=numeric-comp,maxnames=4,sorting=none,url=false,hyperref=true,date=year,doi=true,eprint=true]{biblatex}

\newlist{properties}{enumerate}{10}
\setlist[properties]{label*=\roman*)}
\crefname{propertiesi}{\text{property}}{\text{properties}}
\Crefname{propertiesi}{\text{Property}}{\text{Properties}}
\newlist{steps}{enumerate}{10}
\setlist[steps]{label*=\arabic*)}
\crefname{stepsi}{\text{step}}{\text{steps}}
\Crefname{stepsi}{\text{Step}}{\text{Steps}}
\newlist{points}{enumerate}{10}
\setlist[points]{label*=\arabic*)}
\crefname{pointsi}{\text{point}}{\text{points}}
\Crefname{pointsi}{\text{Point}}{\text{Points}}

\newtheorem{thm}{Theorem}[]
\theoremstyle{definition}

\newtheorem*{crit*}{Criterion}

\newtheorem{lemma}{Lemma}[section]
\newtheorem{remark}{Remark}[section]

\newcommand{\defeq}{\vcentcolon=}
\newcommand{\defeqs}{\stackrel{\scri}{\vcentcolon=}}

\newcommand{\scri}{\mathscr{J}}

\newcommand{\pt}[2]{\tensor{\hat{#1}}{#2}}

\newcommand{\ctru}[3]{\tensor[^{{\,}^{\scriptscriptstyle\!#1\!}\!}]{#2}{#3}}

\newcommand{\ctrdu}[4]{\prescript{{\,}^{\scriptscriptstyle\!#2\!}\!}{{\,}{\scriptscriptstyle\!#1\!}}{\tensor{#3}{#4}}}

\newcommand{\csru}[2]{\tensor[^{{\,}^{\scriptscriptstyle\!#1\!}\!}]{#2}{}}

\newcommand{\ct}[2]{\tensor{{#1}}{#2}}
\newcommand{\pct}[2]{\tensor{\underaccent{*}{#1}}{#2}}
\newcommand{\pctrd}[3]{\tensor[_{{\ }_{\scriptscriptstyle#1}\!}]{\underaccent{*}{#2}}{#3}}
\newcommand{\lct}[2]{\tensor{\accentset{\leftarrow}{#1}}{#2}}

\newcommand{\ctc}[2]{\tensor{\mathring{#1}}{#2}}

\newcommand{\cts}[2]{\tensor{\overline{#1}}{#2}}

\newcommand{\cs}[1]{#1}

\newcommand{\ps}[1]{\hat{#1}}

\newcommand{\df}[1]{\text{d}#1}
\newcommand{\td}[2]{\frac{\df#1}{\df#2}}

\newcommand{\prn}[1]{\left(#1\right)}

\newcommand{\brkt}[1]{\left[#1\right]}

\newcommand{\evalat}[1]{\Big|_{#1}}

\newcommand{\eqs}{\stackrel{\scri}{=}}

\newcommand{\cd}[1]{\tensor{\nabla}{#1}}

\newcommand{\spacef}{\ }

\renewcommand{\P}{\mathcal{P}}

\newcolumntype{M}[1]{>{\centering\arraybackslash}m{#1}}
\newcolumntype{N}{@{}m{0pt}@{}}

\def\be{\begin{equation}}
\def\ee{\end{equation}}
\def\bea{\begin{eqnarray}}
\def\eea{\end{eqnarray}}
\def\bean{\begin{eqnarray*}}
	\def\eean{\end{eqnarray*}}

\newcounter{marginnotecount}[section]

	\title{\Large \textbf{The peeling theorem}\\ \textbf{with arbitrary cosmological constant}}
	\author[]{Francisco Fernández-Álvarez\thanks{francisco.fernandez@ehu.eus}\ }
	\author[]{\ José M. M. Senovilla\thanks{josemm.senovilla@ehu.eus}} 
	\affil[]{Departamento de Física\\ Universidad del País Vasco UPV/EHU\\ Apartado 644, 48080 Bilbao, Spain }
	\date{\today{}}
\begin{document}
	
	\maketitle

	\begin{abstract}
		A method for deriving the asymptotic behaviour of any physical field is presented. This leads to a geometrically meaningful derivation of the peeling properties for arbitrary values of the cosmological constant. Application to the outstanding case of the physical Weyl tensor provides the explicit form of all terms that determine its asymptotic behaviour along arbitrary lightlike geodesics. The results follow under the assumption of a conformal completion {\itshape à la Penrose}. The only freedom available is the choice of a null vector at the conformal boundary of the space-time (which determines the lightlike geodesic arriving there).
	\end{abstract}

	\section{Introduction}\label{sec:introduction}
	
		The asymptotic expansion of the Weyl tensor in powers of an affine parameter along a null geodesic features a very particular algebraic structure when the cosmological constant $ \Lambda $ vanishes. This kind of behaviour is known as the \emph{peeling property} of the Weyl tensor --see, e.g., \cite{Kroon}. It first appeared in the work of Sachs \cite{Sachs61}, a study based on the Petrov classification of the curvature tensor under a certain assumption concerning the `outgoing radiation condition' for vacuum metrics. A subsequent analysis by Sachs \cite{Sachs1962} used a boundary condition à la Sommerfeld in a (Bondi-Sachs \cite{Bondi1962}) metric-based approach. Newman and Penrose also derived the peeling behaviour by the spin-coefficient method in vacuum \cite{Penrose62}, imposing a particular decay condition on the Weyl scalar $ \Psi_{0} $. Alternative logarithmic expansions have been also studied \cite{Winicour85}.\\
		
		The result was soon generalised to arbitrary values of the cosmological constant $ \Lambda $ by Penrose \cite{Penrose65}, in this case by making use of the conformal completion and introducing an auxiliary null vector field parallel-transported along the null geodesic under consideration.  There are much more recent works that aim at a generalisation of the peeling property in the presence of a non-vanishing cosmological constant without using explicitly a conformal completion of the space-time. The results presented in \cite{Saw2017,Saw2020} rely on assumptions that otherwise can be deduced if a conformal completion is used, while the work in \cite{Xie2020} focuses on a particular kind of metrics of the Bondi-Sachs kind and imposes boundary conditions at infinity. \\
		
		The use of null tetrads may shade the interpretation of the peeling property ---since different bases provide with different asymptotic behaviours. In this sense, the relation between the use of different null tetrads in connection with the conformal boundary $ \scri $ of the conformal completion was clearly presented in \cite{Krtous2004}.  \\
		
		In this letter, a derivation of a peeling theorem based on Geroch's ideas \cite{Geroch1977} is given for arbitrary $\Lambda$. Our only assumption is the existence of a conformal completion à la Penrose. We also assume that the conformal factor $\Omega$ can be expanded in powers of an affine parameter $  \lambda $ nearby $\scri$, but this could be appropriately relaxed, if needed, leading to different peeling behaviours. Our derivation of the peeling property is not only valid for \emph{any} value of $ \Lambda $, but also independent of the choice of bases and, furthermore, it provides a general method to determine the asymptotic behaviour of \emph{any} physical field --not only the Weyl tensor. Remarkably, and importantly, it has a  clear geometric meaning.\\
		
		Throughout this work the conformal completion \emph{à la Penrose} $ \prn{M,\ct{g}{_{\alpha\beta}}} $ of a physical space-time $ \prn{\hat{M},\pt{g}{_{\alpha\beta}}} $ satisfying the Einstein field equations
		$$
		\pt{R}{_{\mu\nu}} -\frac{1}{2} \pt{g}{_{\mu\nu}} \hat R +\Lambda \pt{g}{_{\mu\nu}} = \frac{8\pi G}{c^4} \pt{T}{_{\mu\nu}}
		$$
with arbitrary cosmological constant $\Lambda$ is considered, where $\pt{T}{_{\mu\nu}}$ is the energy-momentum tensor. The conformal factor $ \Omega $ relates the two metrics via 
$$
\ct{g}{_{\mu\nu}} \stackrel{\hat M}{=}\Omega^2 \pt{g}{_{\mu\nu}},
$$
see, e.g. \cite{Kroon}, for details. The following conventions are used:  $ (-,+,+,+) $ for the metric signature and $ (\cd{_\alpha}\cd{_\beta}-\cd{_\beta}\cd{_\alpha})\ct{v}{_\gamma} = \ct{R}{_{\alpha\beta\gamma}^\mu} \ct{v}{_\mu}$ for the curvature tensor, where $ \cd{_\alpha} $ is the covariant derivative on $ \prn{\cs{M},\ct{g}{_{\alpha\beta}}}  $. Objects belonging to $ \prn{\hat{M},\pt{g}{_{\alpha\beta}}} $ carry a hat. \\

For convenience, we fix partially the conformal gauge, such that at the conformal boundary $ \scri $ one has $ \cd{_{\alpha}}\ct{N}{_{\beta}}\eqs 0 $, where $ \ct{N}{_{\alpha}}\defeq\cd{_{\alpha}}\Omega $ is a one-form orthogonal to $\{\scri: \Omega=0\}$ and such that $\ct{N}{_{\alpha}}\ct{N}{^{\alpha}}\eqs-\Lambda/3$. Observe that $\ct{N}{^{\alpha}}$ is timelike, null or spacelike for $\Lambda$ positive, zero or negative, respectively. And that $\ct{N}{^{\alpha}}$ points outwards from $\hat M$, along $\scri$, and towards $\hat M$ for $\Lambda$ positive, zero or negative, respectively.  
	\section{The peeling theorem for arbitrary cosmological constant}\label{sec:theorem}
	In \cite{Fernandez-Alvarez_Senovilla-afs} a derivation of the peeling theorem for $ \Lambda=0 $ was presented based on a robust geometrical construction that indeed can be straightforwardly generalised to space-times with  $ \Lambda\neq 0$.\\
	 
	 Consider a future-directed curve $ \gamma\prn{\lambda} $ 
	 		\begin{center}
			\begin{tabular}{  M{1cm} M{1cm}  M{1cm}  M{1cm}  N }
				$ \gamma\prn{\lambda}: $& $ \brkt{-1,0} $&$  \longrightarrow  $ &  $\cs{M} $  & \\ 
				\spacef & $ \lambda $&$ \longrightarrow  $ &  $ p $\spacef.  &\\
			\end{tabular}
		\end{center}
		Here $ \lambda\in[-1,0] $ is the parameter of the curve, and we denote by $ p_{0}\in\scri^+ $ its future endpoint corresponding to $ \lambda_{0}\defeq\lambda\evalat{p_{0}}=0 $ and $ p_{1}\in\ps{M} $ its past endpoint corresponding to $ \lambda_{1}\defeq\lambda\evalat{p_{1}}=-1 $. Choose the parametrisation such that
		\begin{equation}\label{eq:null-geodesic-normalisation}
			\ct{\ell}{^\mu}\ct{N}{_{\mu}}= \frac{\df\Omega}{\df\lambda}
			\eqs -1 \spacef,
		\end{equation}
	where $ \ct{\ell}{^\alpha} $ is the vector field tangent to $ \gamma(\lambda) $. Following the same scheme as in \cite{Fernandez-Alvarez_Senovilla-afs}, which was inspired by \cite{Geroch1977}, let
		\begin{align*}
			\ct{t}{_{\alpha}^\beta}\prn{\lambda_{i},\lambda_{j}}&\spacef \text{ be the parallel propagator w.r.t. }\ct{\Gamma}{^a_{bc}}\spacef \text{ and}\\
			\pt{t}{_{\alpha}^\beta}\prn{\lambda_{i},\lambda_{j}}&\spacef,\text{ the parallel propagator w.r.t. }\pt{\Gamma}{^a_{bc}}\spacef,
		\end{align*}
		where $\ct{\Gamma}{^a_{bc}}$ and $\pt{\Gamma}{^a_{bc}}$ represent the Levi-Civita connection symbols of $\ct{g}{_{\mu\nu}}$ and $\pt{g}{_{\mu\nu}}$, respectively.
	In this construction, three different parallel transports of any covariant tensor field are considered:
		 	\begin{enumerate}
		 	\item From $ p_{0} $ to $ p\prn{\lambda} $ with $ \ct{t}{_{\alpha}^\beta}\prn{\lambda,\lambda_{0}} $,
		 	\item From $ p(\lambda) $ to $ p_{1} $ with $ \pt{t}{_{\alpha}^\beta}\prn{\lambda_{1},\lambda} $,
		 	\item From $ p_{1} $ to $ p_{0} $ with $ \ct{t}{_{\alpha}^\beta}\prn{\lambda_{0},\lambda_{1}} $.
		 	\end{enumerate}
	\begin{figure}[h!]
		 \centering
			\begin{tikzpicture}
				\matrix (m) [matrix of math nodes,row sep=7em,column sep=7em,minimum width=2em]
				{ \Lambda_{p_{0}}M & \Lambda_{p(\lambda)} M\\
				  {} & \Lambda_{p_{1}}M\\};
				\path[-stealth]
				 (m-1-1) edge [left] node [above] {$ \ct{t}{_{\alpha}^\beta}\prn{\lambda,\lambda_{0}} $} (m-1-2)
				 (m-1-2) edge node [right] {$\pt{t}{_{\alpha}^\beta}\prn{\lambda_{1},\lambda}$} (m-2-2)
				   (m-2-2) edge node [left] {$ \ct{t}{_{\alpha}^\beta}\prn{\lambda_{0},\lambda_{1}} \spacef$} (m-1-1);
			\end{tikzpicture}	
	\caption{The parallel propagator $ \ct{t}{_{\alpha}^\beta} $ with respect to the conformal connection constitutes an isomorphism between co-tangent spaces $\Lambda_{p}M  $ at points on $ M $ that belong to $ \gamma $, whereas the parallel propagator $ \pt{t}{_{\alpha}^\beta}  $ with respect to the physical connection is an isomorphism between co-tangent spaces at points on $ \hat M \subset M $ that belong to $ \gamma $, thus \emph{excluding} the endpoint $ p_{0} $ at the boundary $ \scri $. The diagram illustrates the composition of three of these isomorphisms leading to an automorphism on the co-tangent space $ \Lambda_{p_{0}}M $ at $ p_{0}\in\scri $.}\label{fig:diagram-L}
	\end{figure}
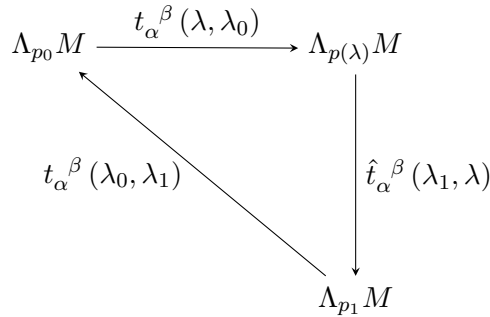
	\Cref{fig:diagram-L} is a diagram showing the sequence of one operation after another: essentially, a transport along $ \gamma $ back and forth, departing from $ \scri^+ $ and interchanging the conformal connection by the physical one for a stretch of $ \gamma $. The three transports, when combined, define an automorphism 
	\begin{center}
		\begin{tabular}{  M{1cm} M{0.8cm}  M{0.8cm}  M{0.8cm}  N }
			$ 	\ctrdu{\lambda}{\gamma}{L}{_{\alpha}^\beta}: $& $ \Lambda_{p_{0}}M $&$ \longrightarrow  $ &  $ \Lambda_{p_{0}}M $   & 
		\end{tabular}
	 	\end{center}
		on the co-tangent space $ \Lambda_{p_{0}} M$ at $ p_{0} $. This automorphism depends on the curve $ \gamma $ and on the intermediate point $ p(\lambda) $ and is explicitly given by
	 	\begin{equation}
	 		\ctrdu{\lambda}{\gamma}{L}{_{\alpha}^\beta}\defeq\ct{t}{_{\alpha}^\mu}\prn{\lambda_{0},\lambda_{1}}\pt{t}{_{\mu}^\rho}\prn{\lambda_{1},\lambda}\ct{t}{_{\rho}^\beta}\prn{\lambda,\lambda_{0}}\spacef .
	 	\end{equation} 
		
		From now on we assume that $\gamma(\lambda)$ is a null geodesic with affine parameter $\lambda$. Observe that such a $\gamma(\lambda)$ is fully determined by $p_0\in\scri$ and an initial null vector $\ct{\ell}{^\alpha}$ of $T_{p_0} M$, therefore all one has to do for the following construction is to pick up any null vector $\ct{\ell}{^\alpha}|_{p_0}$ at the chosen point of $p_0\in\scri$.
	 	Then the characterisation and properties of $ \ctrdu{\lambda}{\gamma}{L}{_{\alpha}^\beta} $ can be found in  \cite{Fernandez-Alvarez_Senovilla-afs}, where it was shown to obey the linear system of ordinary differential equations
		\begin{equation}\label{eq:operatorL-diff-eq}
			\td{	\ctrdu{\lambda}{\gamma}{L}{_{\alpha}^\beta}}{\lambda}=\frac{1}{\Omega\prn{\lambda}}	\ctrdu{\lambda}{\gamma}{L}{_{\alpha}^\mu}\ct{\Lambda}{_{\mu}^\beta}\spacef,
		\end{equation}
where 
			\begin{equation}
			\ctrdu{\lambda}{\gamma}{\Lambda}{_{\beta}^\alpha}=\td{\Omega}{\lambda}\delta^\alpha_{\beta}+\ct{\ell}{^\alpha}\pct{N}{_{\beta}}-\ct{\ell}{_{\beta}}\pct{N}{^\alpha}\spacef,
			\end{equation}
and we have introduced the notation
\begin{equation}
	 		\pct{T}{_{{\alpha_1}...{\alpha_r}}}\defeq \ct{t}{_{\alpha_1}^{\mu_{1}}}\prn{\lambda_{0},\lambda}...\ct{t}{_{\alpha_r}^{\mu_{r}}}\prn{\lambda_{0},\lambda}\ct{T}{_{{\mu_1}...{\mu_r}}}\spacef,
			\end{equation}
for the parallel transport to $p_0$ of arbitrary covariant tensors $ \ct{T}{_{{\alpha_1}...{\alpha_r}}} $ at $ p\prn{\lambda} \in M $. In particular,
	\begin{equation}\label{eq:operatorL-N-transported}
		\pctrd{\lambda}{N}{_{\beta}}\defeq\ct{t}{_{\beta}^\mu}\prn{\lambda_{0},\lambda}\ct{N}{_{\mu}}\prn{\lambda}\spacef,\quad\pctrd{\lambda_0}{N}{_{\beta}}=\ct{N}{_{\beta}}\spacef.
		\end{equation}
	From now on we will drop the $\gamma$ and $\lambda$ decoration from $\ctrdu{\lambda}{\gamma}{L}{_{\beta}^\alpha}$ except in some convenient exceptions, and we introduce yet another notation:
			\begin{equation}
	 		\lct{S}{_{{\alpha_1}...{\alpha_r}}}\defeq \ct{L}{_{\alpha_1}^{\mu_{1}}}...\ct{L}{_{\alpha_r}^{\mu_{r}}}\ct{S}{_{{\mu_1}...{\mu_r}}}
	 	\end{equation}
	for the action of $\ct{L}{_{\alpha}^{\mu}}$ on any covariant $ \ct{S}{_{{\mu_1}...{\mu_r}}} $ at $ p_{0} \in \scri $. 
	One may actually combine the two operations, giving rise to
		\begin{equation}\label{eq:combined}
			\lct{\underaccent{*}{T}}{_{{\alpha_1}...{\alpha_r}}}\defeq\ct{L}{_{\alpha_1}^{\nu_{1}}}...\ct{L}{_{\alpha_r}^{\nu_{r}}} \ct{t}{_{\nu_1}^{\mu_{1}}}\prn{\lambda_{0},\lambda}...\ct{t}{_{\nu_r}^{\mu_{r}}}\prn{\lambda_{0},\lambda}\ct{T}{_{{\mu_1}...{\mu_r}}}\spacef .
		\end{equation}
		
		The two vectors $\ct{\ell}{^\alpha}|_{p_0}$ and $\ct{N}{^\alpha}|_{p_0}$ define a {\em timelike} 2-plane at $p_0\in \scri$, that we call {\em the fundamental plane}. That this plane is timelike is obvious in the cases with $\Lambda \geq 0$, but not so in the case with $\Lambda <0$ because, in principle, there could be null geodesics arriving at $\scri$ and orthogonal to the {\em spacelike} $\ct{N}{^\alpha}$ there, i.e., tangent to $\scri$ at $p_0\in\scri$. However, this is not possible for the null geodesics we need to consider ---that is, coming from $\hat M \subset M$--- due to the particular nature of $\scri$ as a hypersurface in $M$: recall that we have $ \cd{_{\alpha}}\ct{N}{_{\beta}}\eqs 0 $ implying that $\scri$ is a totally geodesic hypersurface, and therefore any geodesic instantly tangent to $\scri$ will remain tangent to $\scri$ all along.\footnote{One might argue that this totally geodesic property depends on the partial gauge fixing that we assumed, but this is not relevant for the result. Had we not fixed the partial gauge, we would nevertheless have $ \cd{_{\alpha}}\ct{N}{_{\beta}}\eqs \frac{1}{4}\cd{_{\rho}}\ct{N}{^{\rho}} \ct{g}{_{\alpha\beta}}$ implying that $\scri$ is a totally umbilical hypersurface in $M$, and therefore {\em null} geodesics will remain tangent to $\scri$ if they are tangent to it at one single point, that is to say, $\scri$ is a ``photon sphere'' when $\Lambda <0$ \cite{Perlick2005}.}	
It follows that there is a {\em unique} future lightlike direction (other than $\ct{\ell}{^\alpha}|_{p_0}$) in the plane defined by $\ct{\ell}{^\alpha}|_{p_0}$ and $\ct{N}{^\alpha}|_{p_0}$ at $p_0\in \scri$ which is determined by the vector
		\begin{equation}\label{eq:kdef}
			\ct{k}{^\alpha}\defeqs \ct{N}{^\alpha}-\frac{\Lambda}{6}\ct{\ell}{^\alpha}\spacef,
			\hspace{1cm} \ct{k}{^\mu}\ct{\ell}{_{\mu}}\eqs -1\spacef,
		\end{equation}
where the second equality is a convenient normalization relation. Notice that $\ct{k}{^\alpha}$ depends on the vector $\ct{\ell}{^\alpha}|_\scri$ ---that is, on the null geodesic $\gamma$--- for $\Lambda \neq 0$, {\em but not for $\Lambda=0$} in which case it is simply $\ct{k}{^\alpha}|_{\Lambda =0} = \ct{N}{^\alpha}$ for {\em every} $\gamma$.\\
		
	Following the same steps as in \cite{Fernandez-Alvarez_Senovilla-afs}, the explicit solution of \eqref{eq:operatorL-diff-eq} 
	can be shown to be
		\begin{equation}\label{eq:operatorL-final-form}
			\ctrdu{\lambda}{\gamma}{L}{_{\beta}^\alpha}=-\ct{k}{^\alpha}\ct{\ell}{_{\beta}}-\frac{1}{2}\ctc{y}{_{\mu}}\ctc{y}{^{\mu}}\ct{\ell}{^\alpha}\ct{\ell}{_{\beta}}-\Xi^2\ct{\ell}{^\alpha}\ct{k}{_{\beta}}+\Xi\ctc{P}{^\alpha_{\beta}}-\Xi\ct{\ell}{^\alpha}\ctc{y}{_{\beta}}+\ctc{y}{^\alpha}\ctc{\ell}{_{\beta}}\spacef,
		\end{equation}
where 
\begin{itemize}
\item the spacelike vector $ \ctc{y}{^\alpha}\in  T_{p_0}M$ is explicitly given by
		\begin{equation}\label{eq:operatorL-y-solution}
		\ctc{y}{^\alpha}=-\Omega\prn{\lambda}\int_{\lambda_{1}}^{\lambda}\frac{1}{\Omega^2\prn{\lambda'}}\pctrd{\lambda'}{N}{^{\alpha}}\df{\lambda'}\spacef,
		\end{equation}
and is orthogonal to the fundamental timelike plane 
		\begin{align}
			\ctc{y}{^\mu}\ct{\ell}{_{\mu}}\evalat{p_{0}}=0=\ctc{y}{^\mu}\ct{k}{_{\mu}}\evalat{p_{0}}\spacef.
		\end{align}
\item $ \ctc{P}{^\alpha_{\beta}} \defeqs \ct{\delta}{^\alpha_{\beta}} +\ct{k}{^\alpha}\ct{\ell}{_{\beta}}+\ct{\ell}{^\alpha}\ct{k}{_{\beta}}$ is the projector orthogonal to the fundamental plane
\item and we use the shorthand
		\begin{equation}\label{eq:operatorL-xi}
	 		\cs{\Xi}\defeq \frac{\Omega(\lambda)}{\Omega(\lambda_1)}\spacef.
	 	\end{equation}
\end{itemize}
Observe that $ \ctc{y}{^\alpha} $, $ \ctc{P}{^\alpha_{\beta}} $ and $ \Xi $ depend on the null geodesic $ \gamma $ and/or on the point $ p\prn{\lambda} $. 

An equivalent useful expression for $ \ctrdu{\lambda}{\gamma}{L}{_{\alpha}^\beta} $ using $\ct{N}{^\alpha}|_{p_0}$ instead of $\ct{k}{^\alpha}$ is
\begin{equation}\label{eq:operatorL-final-form2}
			\ctrdu{\lambda}{\gamma}{L}{_{\beta}^\alpha}=-\ct{N}{^\alpha}\ct{\ell}{_{\beta}}+\left(\frac{\Lambda}{3}(1+\Xi^2)-\frac{1}{2}\ctc{y}{_{\mu}}\ctc{y}{^{\mu}}\right)\ct{\ell}{^\alpha}\ct{\ell}{_{\beta}}-\Xi^2\ct{\ell}{^\alpha}\ct{N}{_{\beta}}+\Xi\ctc{P}{^\alpha_{\beta}}-\Xi\ct{\ell}{^\alpha}\ctc{y}{_{\beta}}+\ctc{y}{^\alpha}\ctc{\ell}{_{\beta}}
		\end{equation}
valid for arbitrary values of $\Lambda\,  (\eqs -3 \ct{N}{_\alpha}\ct{N}{^\alpha})$.\\


	As explained in \cite{Fernandez-Alvarez_Senovilla-afs}, the asymptotic behaviour of \emph{any} physical field ---here represented by the fully covariant version $\ct{T}{_{{\mu_1}...{\mu_r}}}$ of its defining tensor field --- can be obtained by applying the next steps
		\begin{steps}
		\item Construct $\lct{\underaccent{*}{T}}{_{{\alpha_1}...{\alpha_r}}}$ from $\ct{T}{_{{\mu_1}...{\mu_r}}}|_{p(\lambda)}$ according to \eqref{eq:combined},\label{it:asymptotic-propagation-first-step}
			\item Expand the resulting expression in powers of $ \lambda $ near $ \lambda_{0}=0 $. \label{it:asymptotic-propagation-third-step}
		\end{steps}
	This arrangement  `compares' the parallel propagation of $\ct{T}{_{{\mu_1}...{\mu_r}}}$ in $ \ps{M} $ from the point $ p_1 $ to $ p(\lambda) $ with its parallel propagation in $ \cs{M} $ between the same two points. Pushing $ \lambda $ towards the limit value $ \lambda_0=0  $ one takes this comparison towards infinity of $ \ps{M} $.
		\begin{table}[h!]
									\centering
									\begin{tabular}{ |M{3cm}| M{3cm} | M{3cm}| M{3cm}| N }
										\hline
											\quad Weyl-tensor candidate & Non-vanishing $ \csru{(a)}{\psi}_{i} $ &  $ \csru{(a)}{\psi}_{i} $ when $ \lambda=\lambda_{0} $  & PND &\\ \hline 
										$ \ctru{(4)}{C}{_{\alpha\beta\gamma\delta}} $& $ \csru{(4)}{\psi}_{4} $ & $ \Omega^{-2}(\lambda_{1}) \cs{\phi}_{4} $& $ \prn{\ct{\ell}{^\alpha},\ct{\ell}{^\alpha},\ct{\ell}{^\alpha},\ct{\ell}{^\alpha}} $&\\[1cm] \hline 
										$ \ctru{(3)}{C}{_{\alpha\beta\gamma\delta}} $& $ \csru{(3)}{\psi}_{3} $ & $ \Omega^{-3}(\lambda_{1}) \cs{\phi}_{3} $&$ \prn{\ct{\ell}{^\alpha},\ct{\ell}{^\alpha},\ct{\ell}{^\alpha},\ct{k}{^\alpha}} $&\\[1cm] \hline
										$ \ctru{(2)}{C}{_{\alpha\beta\gamma\delta}} $& $ \csru{(2)}{\psi}_{2} $ & $ \Omega^{-4}(\lambda_{1}) \cs{\phi}_{2} $&$ \prn{\ct{\ell}{^\alpha},\ct{\ell}{^\alpha},\ct{k}{^\alpha},\ct{k}{^\alpha}} $&\\[1cm] \hline				
										$ \ctru{(1)}{C}{_{\alpha\beta\gamma\delta}} $& $ \csru{(1)}{\psi}_{1} $ & $ \Omega^{-5}(\lambda_{1}) \cs{\phi}_{1} $&$ \prn{\ct{\ell}{^\alpha},\ct{k}{^\alpha},\ct{k}{^\alpha},\ct{k}{^\alpha}} $&\\[1cm] \hline
										$ \ctru{(0)}{C}{_{\alpha\beta\gamma\delta}} $& $ \csru{(0)}{\psi}_{0} $ & $ \Omega^{-6}(\lambda_{1}) \cs{\phi}_{0} $&$ \prn{\ct{k}{^\alpha},\ct{k}{^\alpha},\ct{k}{^\alpha},\ct{k}{^\alpha}} $&\\[1cm] \hline	
									\end{tabular}
	\caption[Asymptotic propagation of the Weyl tensor]{The asymptotic behaviour of the physical Weyl tensor \eqref{eq:asymptotic-propagation-weyl-tensor} is composed by five terms as listed above, each one with the symmetries of a Weyl tensor, two of them having Petrov type N, two others Petrov type III, and one more with Petrov type D. In a null tetrad containing $\{\ct{\ell}{^\alpha},\ct{k}{^\alpha}\}$ as the two null vectors, each of the five terms possesses a unique non-vanishing Weyl scalar, as shown, which in the limit $ \lambda\rightarrow\lambda_{0}=0 $ coincides up to a multiplicative constant with one of the scalars of the rescaled Weyl tensor $ \ct{d}{_{\alpha\beta\gamma}^\delta} $ (here denoted by $\phi_a$). The corresponding principal null directions (PND) are listed in the last column.
									}
									\label{tab:weyl-candidates}
								\end{table}
								\\
								
	This program can be applied to the physical Weyl tensor $ \pt{C}{_{\alpha\beta\gamma\delta}} $. First, consider $ \pt{C}{_{\alpha\beta\gamma\delta}} $ at $ p(\lambda) $ and take \cref{it:asymptotic-propagation-first-step} to define the tensor $\lct{\underaccent{*}{\hat{C}}}{_{\alpha\beta\gamma\delta}}$ at $ p_{0} $ which,  upon use of \eqref{eq:operatorL-final-form} can be proven to be composed of five terms
	\begin{equation}\label{eq:asymptotic-propagation-weyl-tensor}
						 \lct{\underaccent{*}{\hat{C}}}{_{\alpha\beta\gamma\delta}}=\Omega \ctru{(4)}{C}{_{\alpha\beta\gamma\delta}}+\Omega^2\ctru{(3)}{C}{_{\alpha\beta\gamma\delta}}+\Omega^3\ctru{(2)}{C}{_{\alpha\beta\gamma\delta}}+\Omega^4\ctru{(1)}{C}{_{\alpha\beta\gamma\delta}}+\Omega^5\ctru{(0)}{C}{_{\alpha\beta\gamma\delta}},
					\end{equation}
				where $ \ctru{(a)}{C}{_{\alpha\beta\gamma\delta}} $ with $ a=0,1,2,3,4 $ are Weyl-tensor candidates\footnote{A Weyl-tensor candidate is any tensor with all the algebraic properties of the Weyl tensor}, regular in the limit to $ \lambda_{0}=0 $ and with the algebraic properties listed in \cref{tab:weyl-candidates}. These tensors as well as $\lct{\underaccent{*}{\hat{C}}}{_{\alpha\beta\gamma\delta}}$ depend on $ \lambda $ via $\Omega(\lambda)$. Assuming that $\Omega$ admits an expansion in powers\footnote{This assumption could be relaxed and, for more general dependences on $\lambda$, the corresponding generalized peeling behaviours can be equally produced.}  of $\lambda$ they can be expanded around $ \lambda_{0}=0 $ as
					\begin{equation}\label{eq:weyl-terms-expanded}
						\ctru{(a)}{C}{_{\alpha\beta\gamma\delta}}=\ctru{(a,0)}{C}{_{\alpha\beta\gamma\delta}}+\sum_{i=1}^{\infty}\ctru{(a,i)}{C}{_{\alpha\beta\gamma\delta}}\lambda^i\spacef.
					\end{equation}
	It has to be pointed out that the leading-order term of \cref{eq:asymptotic-propagation-weyl-tensor} reads
		 	\begin{equation}\label{eq:asymptotic-leading-term-weyl}
				\ctru{(4,0)}{C}{_{\alpha\beta\gamma\delta}}=\frac{4}{\Omega^2\prn{\lambda_{1}}}\ct{d}{_{\mu\nu\rho\sigma}}\ct{k}{^\nu}\ct{k}{^\sigma}\ctc{P}{_{[\alpha}^\mu}\ct{\ell}{_{\beta]}}\ctc{P}{_{[\gamma}^\rho}\ct{\ell}{_{\delta]}}
			\end{equation}
		and is determined by the rescaled Weyl tensor $ \ct{d}{_{\alpha\beta\gamma}^\delta}$, defined by $ \Omega\ct{d}{_{\alpha\beta\gamma}^\delta}\defeq \ct{C}{_{\alpha\beta\gamma}^\delta}\stackrel{\hat M}{=} \pt{C}{_{\alpha\beta\gamma}^\delta} $,  projected to a Petrov-type N Weyl-candidate tensor at $p_0\in \scri$. \\

	Now, the peeling theorem can be presented.
			\begin{thm}[Peeling of the Weyl tensor for arbitrary $ \Lambda $]\label{thm:peeling-weyl}
					Let $ \prn{\cs{M},\ct{g}{_{\alpha\beta}}} $ be a conformal completion \emph{à la} Penrose of a physical space-time with arbitrary $ \Lambda $ and choose an arbitrary point $p_0\in \scri^+$ and a future null vector $\ct{\ell}{^\alpha}$ at $p_0$ normalized according to \cref{eq:null-geodesic-normalisation}. Denote by $ \gamma $ the lightlike geodesic with tangent vector $\ct{\ell}{^\alpha}$ at $p_0$, and let $\lambda$ be its affine parameter such that $ \lambda|_{p_0}=\lambda_{0}=0 $. 
					Then, assuming that $\Omega$ admits an expansion in powers of $\lambda$, the asymptotic behaviour of the physical Weyl tensor $ \pt{C}{_{\alpha\beta\gamma\delta}} $ along $ \gamma $ follows by application of \cref{it:asymptotic-propagation-first-step,it:asymptotic-propagation-third-step} on page \pageref{it:asymptotic-propagation-first-step} and reads
						\begin{equation}\label{eq:peeling-weyl}
 						 \lct{\underaccent{*}{\hat{C}}}{_{\alpha\beta\gamma\delta}}=	\lambda \ctru{(N)}{d}{_{\alpha\beta\gamma\delta}}+\lambda^2\ctru{(III)}{e}{_{\alpha\beta\gamma\delta}}+\lambda^3\ctru{(II/D)}{f}{_{\alpha\beta\gamma\delta}}+\lambda^4\ctru{(I)}{g}{_{\alpha\beta\gamma\delta}}+\lambda^5\ctru{(I)}{h}{_{\alpha\beta\gamma\delta}}+\mathcal{O}\prn{\lambda^6}\spacef,
						\end{equation}
						near $ \lambda=\lambda_{0}=0 $, where the tensors
						\begin{align}
							\ctru{(N)}{d}{_{\alpha\beta\gamma\delta}}\defeq&- \ctru{(4,0)}{C}{_{\alpha\beta\gamma\delta}}\spacef,\label{eq:rrw-d}\\
							\ctru{(III)}{e}{_{\alpha\beta\gamma\delta}}\defeq& \ctru{(3,0)}{C}{_{\alpha\beta\gamma\delta}}+\frac{\Omega_{2}}{2}\ctru{(4,0)}{C}{_{\alpha\beta\gamma\delta}}-\ctru{(4,1)}{C}{_{\alpha\beta\gamma\delta}}\spacef,\label{eq:rrw-e}\\
							\ctru{(II/D)}{f}{_{\alpha\beta\gamma\delta}}\defeq& -\ctru{(2,0)}{C}{_{\alpha\beta\gamma\delta}}-\Omega_{2}\ctru{(3,0)}{C}{_{\alpha\beta\gamma\delta}}+\frac{\Omega_{3}}{6}\ctru{(4,0)}{C}{_{\alpha\beta\gamma\delta}}+\ctru{(3,1)}{C}{_{\alpha\beta\gamma\delta}}-\ctru{(4,2)}{C}{_{\alpha\beta\gamma\delta}}\nonumber\\
							&+\frac{\Omega_{2}}{2}\ctru{(4,1)}{C}{_{\alpha\beta\gamma\delta}}\spacef,\label{eq:rrw-f}\\
							\ctru{(I)}{g}{_{\alpha\beta\gamma\delta}}\defeq& \ctru{(1,0)}{C}{_{\alpha\beta\gamma\delta}}+\frac{3\Omega_{2}}{2}\ctru{(2,0)}{C}{_{\alpha\beta\gamma\delta}}+\prn{\frac{\Omega_{2}^2}{4}-\frac{\Omega_{3}}{3}}\ctru{(3,0)}{C}{_{\alpha\beta\gamma\delta}}+\frac{\Omega_{4}}{4!}\ctru{(4,0)}{C}{_{\alpha\beta\gamma\delta}}\nonumber\\
							&-\ctru{(2,1)}{C}{_{\alpha\beta\gamma\delta}}+\ctru{(3,2)}{C}{_{\alpha\beta\gamma\delta}}-\ctru{(4,3)}{C}{_{\alpha\beta\gamma\delta}}+\frac{\Omega_{2}}{2}\ctru{(4,2)}{C}{_{\alpha\beta\gamma\delta}}-\Omega_{2}\ctru{(3,1)}{C}{_{\alpha\beta\gamma\delta}}\nonumber\\
							&+\frac{\Omega_{3}}{6}\ctru{(4,1)}{C}{_{\alpha\beta\gamma\delta}}\spacef,\label{eq:rrw-g}\\
							\ctru{(I)}{h}{_{\alpha\beta\gamma\delta}}\defeq& -\ctru{(0,0)}{C}{_{\alpha\beta\gamma\delta}}-2\Omega_{2}\ctru{(1,0)}{C}{_{\alpha\beta\gamma\delta}}+\prn{\frac{\Omega_{3}}{2}-\frac{3}{4}\Omega_{2}^2}\ctru{(2,0)}{C}{_{\alpha\beta\gamma\delta}}+\prn{\frac{1}{6}\Omega_{2}\Omega_{3}-\frac{\Omega_{4}}{12}}\ctru{(3,0)}{C}{_{\alpha\beta\gamma\delta}}\nonumber\\
							&+\frac{\Omega_{5}}{5!}\ctru{(4,0)}{C}{_{\alpha\beta\gamma\delta}}+\ctru{(1,1)}{C}{_{\alpha\beta\gamma\delta}}-\ctru{(2,2)}{C}{_{\alpha\beta\gamma\delta}}+\ctru{(3,3)}{C}{_{\alpha\beta\gamma\delta}}-\ctru{(4,4)}{C}{_{\alpha\beta\gamma\delta}}\nonumber\\
							&+\frac{3}{2}\Omega_{2}\ctru{(2,1)}{C}{_{\alpha\beta\gamma\delta}}-\Omega_{2}\ctru{(3,2)}{C}{_{\alpha\beta\gamma\delta}}+\frac{\Omega_{2}}{2}\ctru{(4,3)}{C}{_{\alpha\beta\gamma\delta}}+\frac{\Omega_{2}^2}{4}\ctru{(3,1)}{C}{_{\alpha\beta\gamma\delta}}\nonumber\\
							&-\frac{\Omega_{3}}{3}\ctru{(3,1)}{C}{_{\alpha\beta\gamma\delta}}+\frac{\Omega_{3}}{6}\ctru{(4,2)}{C}{_{\alpha\beta\gamma\delta}}+\frac{\Omega_{4}}{4!}\ctru{(4,1)}{C}{_{\alpha\beta\gamma\delta}}\spacef,\label{eq:rrw-h}
						\end{align}
					are Weyl-tensor candidates labelled with their {\em generic} Petrov type, respectively; $ \Omega_{i} $, with $ i=1,2,3,4,5 $, is the $ i$-th derivative of $ \Omega $ w.r.t. $ \lambda $ evaluated at $ \lambda=\lambda_{0}=0 $, and $ \ctru{(a,b)}{C}{_{\alpha\beta\gamma\delta}} $ with $ a=0,1,2,3,4 $, $ b=0,1,2,... $ are the tensors appearing in the expansion \eqref{eq:weyl-terms-expanded} of the Weyl-tensor candidates $ \ctru{(a)}{C}{_{\alpha\beta\gamma\delta}} $ in \eqref{eq:asymptotic-propagation-weyl-tensor}, also shown in \cref{tab:weyl-candidates}, each one having only one non-vanishing Weyl scalar $ \csru{(a)}{\psi}_{a} $ in the null tetrad containing $ \ct{\ell}{^\alpha} $ and $ \ct{k}{^\alpha} $.
				\end{thm}
				\begin{remark}
					There is only one directional freedom in this result: the choice of the null vector $\ct{\ell}{^\alpha}$ at $p_0$ that determines the geodesic $ \gamma $. 
					Once 
		$\ct{\ell}{^\alpha}|_{p_0}$ is chosen, the null vector $ \ct{k}{^\alpha} $ is automatically fixed: each $\ct{\ell}{^\alpha}|_{p_0}$ --equivalently, each null geodesic arriving at $p_0$--- has a uniquely associated $ \ct{k}{^\alpha} $ and the only special feature of the $ \Lambda=0 $ case is that $ \ct{k}{^\alpha}=\ct{N}{^\alpha} $ for every possible $\ct{\ell}{^\alpha}|_{p_0}$. Hence, {\em in the case with $\Lambda=0$}, if the leading term of \cref{eq:asymptotic-leading-term-weyl} vanishes for one null geodesic reaching $ p_{0} $, it does vanish for every possible null geodesic arriving at $ p_{0} $. Moreover and for the same reason, in that case if $  \ctru{(III)}{e}{_{\alpha\beta\gamma}^\delta}$ has Petrov type N or zero for one null geodesic reaching $ p_{0} $, the same holds true for every possible null geodesic arriving at $ p_{0} $. 
				\end{remark}
				\begin{remark}
					In \cite{Krtous2004} a result was given for the decay of the physical Weyl scalar $ \ps{\psi}_{4} $ (expressed in powers of a physical affine parameter) that exhibits a strong directional dependence in the $ \Lambda\neq0 $ scenarios. However, from the point of view of \cref{thm:peeling-weyl} the important feature is not just the existence of a directional dependence in the components of the physical Weyl tensor when propagated towards infinity (that is, the dependence on $ \gamma $), but the general algebraic structure that these components have, which does not depend on the null geodesic $ \gamma $ nor on the value of $ \Lambda $. This does not excludes that, in the $ \Lambda\neq0 $ case, for particular $ \gamma $ some components may be more degenerate in the Petrov-classification sense --in concordance with the discussion in \cite{Krtous2004}, see also \cite{Penrose1986}. Moreover, when studying the radiative components of the gravitational field, it is essential to take into account not only $ 	\ctru{(N)}{d}{_{\alpha\beta\gamma\delta}} $ (or $ \phi_{4} $) but also 	$ \ctru{(III)}{e}{_{\alpha\beta\gamma\delta}} $ --this is further explained next, see also \cite{Fernandez-Alvarez_Senovilla-afs}.
				\end{remark}
					
		In \cite{Fernandez-Alvarez_Senovilla-afs} the relationship between the peeling theorem and the no radiation condition for $ \Lambda=0 $ was found. Although much weaker, a possible interesting relation can be established for the $ \Lambda>0 $ scenario using a `no incoming radiation' condition put forward in \cite{Fernandez-Alvarez_Senovilla-dS}.
		To show this, and once the null $\ct{\ell}{^\alpha}|_{p_0}$ 
		has been chosen, the question is to know when gravitational radiation arrives at $p_0$ via the null geodesic $\gamma$ detemined by $\ct{\ell}{^\alpha}|_{p_0}$.\\
		
		The no incoming radiation criterion can be expressed in several forms, here we just use its simplest form which is well adapted to our goal: incoming radiation propagating along $  \ct{m}{^\alpha}  $ on $ \scri^+ $ is absent if $\ct{k}{^\alpha}$ is a repeated principal null direction of the rescaled Weyl tensor $ \ct{d}{_{\alpha\beta\gamma}^\delta} $. Here $ \ct{m}{^\alpha} $ is the unique unit spacelike vector lying in the fundamental timelike plane, 
		tangent to $ \scri^+ $ (i.e., $ \ct{m}{_\alpha} \ct{N}{^\alpha}  |_{p_0}=0$) and pointing coincidentally with $ \ct{\ell}{^\alpha} |_{p_0}$ (i.e. $ \ct{m}{_\alpha} \ct{\ell}{^\alpha} |_{p_0}>0$).\\

		The relationship with the peeling theorem for $ \Lambda>0 $ follows from the next observation
			\begin{lemma}\label{thm:noincoming-peeling}
				   $ \ct{k}{^\alpha}\stackrel{p_0}{=} \ct{N}{^\alpha}-(\Lambda/6)\ct{\ell}{^\alpha} $ is a repeated principal null direction of $ \ct{d}{_{\alpha\beta\gamma}^\delta}|_{p_0} $ if and only if $ \ctru{(N)}{d}{_{\alpha\beta\gamma}^\delta}= 0 $ and $  \ctru{(III)}{e}{_{\alpha\beta\gamma}^\delta}$ has Petrov type N or zero for the null geodesic $ \gamma $ reaching $p_0 $ with tangent vector $ \ct{\ell}{^\alpha}|_{p_0} $. 
			\end{lemma}
			\begin{proof}
The result is easily shown recalling \cref{eq:asymptotic-leading-term-weyl} and noting that $ \ctru{(3,0)}{C}{_{\alpha\beta\gamma\delta}} $ is the projection to a type-III Weyl-tensor candidate of $ \ct{d}{_{\alpha\beta\gamma}^\delta} $ such that $ \ctru{(3,0)}{C}{_{\alpha\beta\gamma\delta}}=0 $ if and only if $ \ctru{(III)}{e}{_{\alpha\beta\gamma}^\delta}$ has Petrov type N or zero -- see remark V.3 in \cite{Fernandez-Alvarez_Senovilla-afs}.
			\end{proof}
			\begin{remark}
				Typically, the leading term $ \ctru{(N)}{d}{_{\alpha\beta\gamma}^\delta} $ is treated as the `radiative' component of the gravitational field. However, this is not accurate because $ \ctru{(III)}{e}{_{\alpha\beta\gamma}^\delta} $ must also be taken into account in order to determine the presence/absence of radiation already in the $ \Lambda=0 $ scenario  \cite{Fernandez-Alvarez_Senovilla-afs} --this is in agreement with \cite{Sachs1962}. What has just been shown is that a similar situation arises with the absence of incoming radiation condition along $\ct{m}{^\alpha}$ in the $ \Lambda>0 $ case. 
			\end{remark}
			
			A particular situation of interest arises if, once the null vector $ \ct{\ell}{^\alpha}|_{p_0} $ is chosen and therefore $\ct{k}{^\alpha}$ is uniquely determined, we ask about the peeling behaviour for the null geodesic arriving at $p_0$ {\em precisely} with $\ct{k}{^\alpha}$ as tangent vector. Let $ \ctru{(N)}{{d'}}{_{\alpha\beta\gamma}^\delta} $ and $ \ctru{(III)}{{e'}}{_{\alpha\beta\gamma}^\delta} $ denote the corresponding Weyl-tensor candidates associated to $\ct{k}{^\alpha}$ as described in \cref{thm:peeling-weyl}. Then,
			 	\begin{equation}\label{eq:P=0}
			 		   \ctru{(N)}{d}{_{\alpha\beta\gamma}^\delta}=\ctru{(N)}{{d'}}{_{\alpha\beta\gamma}^\delta}= 0  \text{ and }   \ctru{(III)}{e}{_{\alpha\beta\gamma}^\delta}\text{ and }\ctru{(III)}{{e'}}{_{\alpha\beta\gamma}^\delta} \text{ have Petrov type N or zero}\implies\cts{P}{^a}\evalat{p_{0}}=0
			 	\end{equation}
where $ \cts{\P}{^a} $ is the asymptotic super-Poynting vector field \cite{Fernandez-Alvarez_Senovilla20b} whose vanishing determines the absence of gravitational radiation arriving at $\scri^+$.\\
			
			To prove \eqref{eq:P=0} notice that the assumptions are equivalent to the vanishing of the two radiant supermomenta associated to $ \ct{\ell}{^\alpha} $ and $ \ct{k}{^\alpha} $, respectively --one shows this in a similar way as the proof of \cref{thm:noincoming-peeling}, using the properties of radiant supermomenta presented in \cite{Fernandez-Alvarez_Senovilla-dS}. But this implies the vanishing of $ \cts{P}{^a} $ \cite{Fernandez-Alvarez_Senovilla-dS} .
			\begin{remark}
				The double implication requires a more elaborated condition deductible from Lemma II.4 of \cite{Fernandez-Alvarez_Senovilla-dS}.
			\end{remark}
		\subsubsection*{Acknowledgments}
		Work supported under Grants No. FIS2017-85076-P (Spanish MINECO/AEI/FEDER, EU) and No. IT956-16 (Basque Government). 
		
\printbibliography
\end{document}